\documentclass[conference,a4paper]{IEEEtran}

\usepackage{graphicx}
\usepackage{amsmath,mathrsfs}
\usepackage{amsthm}
\usepackage{amssymb}  
\usepackage{amsbsy}
\usepackage{color}
\usepackage{stfloats}
\usepackage{algorithm}
\usepackage{algpseudocode}
\usepackage{flushend}

\newtheoremstyle{note}
{3pt}
{1pt}
{}
{\parindent}
{\itshape}
{:}
{.5em}
{\thmname{#1}\thmnumber{ #2}\thmnote{\thmnote{ (#3)}}}
\theoremstyle{note}

\newtheorem{theorem}{Theorem}
\newtheorem{lemma}{Lemma}
\newtheorem{remark}{Remark}

\begin{document}
\sloppy
\title{Polar Lattices: Where Ar{\i}kan Meets Forney}

\author{\IEEEauthorblockN{}
\IEEEauthorblockA{\\
\\
\\
Email: \{y.yan10, c.ling\}@imperial.ac.uk}
\and
\IEEEauthorblockN{Xiaofu Wu}
\IEEEauthorblockA{Nanjing University of Posts and Telecommunications\\
Nanjing, China\\
Email: xfuwu@ieee.org}
}
\author{
  \IEEEauthorblockN{Yanfei Yan and Cong Ling}
  \IEEEauthorblockA{Department of Electrical and Electronic Engineering\\
    Imperial College London\\
    London, UK\\
    Email: \{y.yan10, c.ling\}@imperial.ac.uk}
  \and
  \IEEEauthorblockN{Xiaofu Wu}
  \IEEEauthorblockA{Nanjing University of Posts and Telecommunications\\
    Nanjing, China\\
    Email: xfuwu@ieee.org}
}



\maketitle

\begin{abstract}
In this paper, we propose the explicit construction of a new class of lattices based on polar codes, which are provably good for the additive white Gaussian noise (AWGN) channel. We follow the multilevel construction of Forney \textit{et al.} (i.e., Construction D), where the code on each level is a capacity-achieving polar code for that level. The proposed polar lattices are efficiently decodable by using multistage decoding. Performance bounds are derived to measure the gap to the generalized capacity at given error probability. A design example is presented to demonstrate the performance of polar lattices.
\end{abstract}

\section{Introduction}
%
%
%
%

Polar codes, proposed by Ar{\i}kan in \cite{polarcodes}, can
provably achieve the capacity of binary memoryless
symmetric (BMS) channels. The authors in \cite{polarconstruction} and \cite{Ido}
proposed efficient algorithms to construct polar codes for BMS
channels with almost linear complexity. An attempt (from the theoretic side) to construct polar codes for the additive white Gaussian noise (AWGN) channel was based on the technique for multiple- access channels \cite{Abbe_tran}, \cite{Abbe}. Until now, however it is still an open problem to
construct practical polar codes to achieve the capacity of the AWGN channel.

Lattice codes are the counterpart of linear codes in the Euclidean space that are known to achieve the AWGN channel capacity, based on the random coding argument \cite{zamir}. At the core of lattice coding is an \emph{AWGN-good} lattice, whose error probability, for infinite lattice coding, vanishes as long as the volume-to-noise ratio (VNR) is larger than 1. Such a lattice is also referred to as achieving the generalized capacity \cite{poltyrev} or the sphere bound \cite{forney6}. Several practical lattice constructions have been introduced recently, mostly inspired by low-density parity check (LDPC) codes and turbo codes \cite{ldpclattice,IntegerLDLC,ldlc,turbolattice}. However, there is no proof that such lattices are AWGN-good. This paper will focus on the the design of the AWGN-good lattice, while the issue of a finite power constraint can be dealt with by the standard shaping technique (see, e.g., \cite{zamir}).

In this paper, we propose polar lattices to solve both problems aforementioned. Polar lattices are explicit AWGN-good lattices constructed from polar codes. More precisely, we use Forney \textit{et al.}'s multilevel construction \cite{forney6}, where for each level we build a polar code to achieve its capacity. The analysis of \cite{forney6} shows that if the parameters of the component lattices are carefully chosen and if the component
codes on all levels are capacity-achieving, then the resulted lattice will achieve the sphere bound.

This paper is built on the basis of our prior attempt \cite{Yan}, where we followed the structure of Barnes-Wall lattices, but changed the Reed-Muller code to a polar code on each level. By applying the \emph{capacity rule} \cite{forney6} rather than the Barnes-Wall rule, the performance of our polar lattices has been greatly improved, as shown in Fig. \ref{fig:1D_bw}. The encoding and decoding complexity of polar lattices is almost the same as that of Barnes-Wall lattices. We use the same multistage decoder for both polar lattices and Barnes-Wall lattices. The reason for the relatively poor performance of Barnes-Wall lattices is that they violate the capacity rule: the rates of some component codes exceed the channel capacities of the corresponding levels. More discussion will be given later.

\begin{figure}
    \centering
    \includegraphics[width=8.9cm]{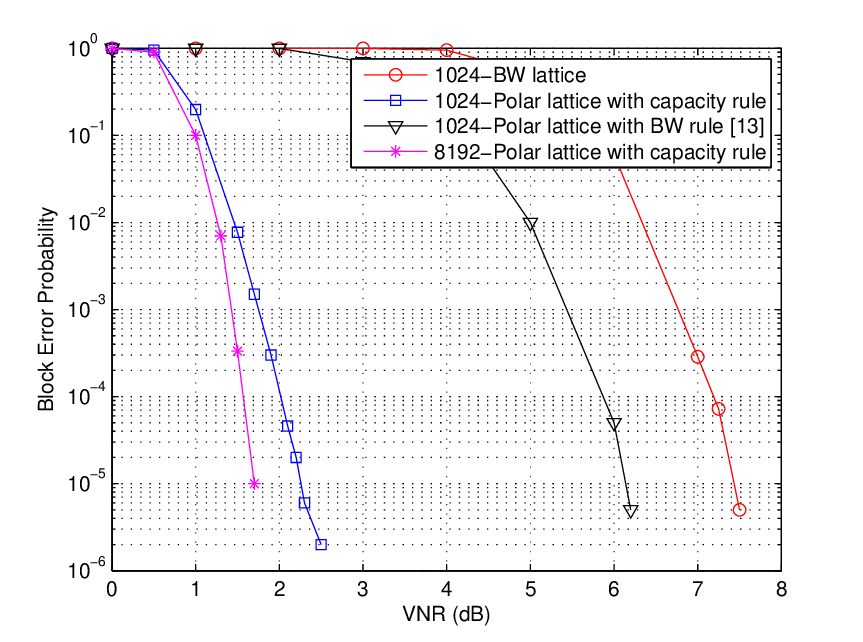}
    \vspace{-12pt}
    \caption{Block error probabilities of polar lattices with length $N=1024$ and $N=8192$ under multistage decoding.}
    \vspace{-12pt}
    \label{fig:1D_bw}
\end{figure}

At the final stage of this paper, we became aware of a related work \cite{multilevelpolar} where polar codes were used in multilevel coded modulation. However, \cite{multilevelpolar} is not about the construction of lattices, and no bounds on the error performance were derived.

Our paper is organized as follows.
Section II presents the background of lattice theory and multilevel construction. The construction of polar codes for the mod-2 BAWGN channel, which serves as a building block for polar lattices, is given in Section III. In Section IV,
we derive the performance bounds and prove that polar lattices can achieve the sphere bound. The practical design and simulation results are presented in section V. We conclude the paper in Section VI.

\section{Background on Lattices and Multilevel Construction}
A lattice is a discrete subgroup of $\mathbb{R}^{n}$ which can be described in terms of a generator matrix:
\begin{eqnarray}
\Lambda=\{\lambda=Gx:x\in\mathbb{Z}^{n}\}, \notag\
\end{eqnarray}
where $G$ is an $n$-by-$n$ full-rank real matrix.

Let $V(\Lambda)$ be the fundamental volume of $\Lambda$, and $\sigma^2$ be the variance of the AWGN.
The VNR of an $n$-dimensional lattice $\Lambda$ is defined as
\begin{eqnarray}
\alpha^2(\Lambda,\sigma^2)=\frac{V(\Lambda)^\frac{2}{n}}{2\pi e\sigma^2}.
\label{eqn:VNR}
\end{eqnarray}

A sublattice $\Lambda'$ induces a partition (denoted by $\Lambda/\Lambda'$) of $\Lambda$ into equivalence groups modulo $\Lambda'$. Let $\Lambda_{1}/\cdots/\Lambda_{r-1}/\Lambda_{r}$ be an $n$-dimensional lattice partition chain. In the multilevel construction of lattices, for each
partition $\Lambda_{\ell}/\Lambda_{\ell+1}$ ($1\leq \ell \leq r-1$) a code $\mathcal{C}_{\ell}$ over $\Lambda_{\ell}/\Lambda_{\ell+1}$
selects a sequence of cosets of $\Lambda_{\ell+1}$. In this paper, we use the one-dimensional lattice partition chain $\mathbb{Z}/2\mathbb{Z}/\cdot\cdot\cdot/2^{r-1}\mathbb{Z}$. This is known as ``Construction D".

Let $\psi$ be the natural embedding of $\mathbb{F}_{2}^{N}$ into $\mathbb{Z}^{N}$, where $\mathbb{F}_{2}$ is the binary field. Let $\mathcal{C}_{1}\subseteq\mathcal{C}_{2}\cdot\cdot\cdot\subseteq\mathcal{C}_{r-1}$ be a family of nested binary linear codes and let $k_{\ell}=\text{dim}(\mathcal{C}_{\ell})$ for $1\leq\ell\leq r-1$. Let $b_{1}, b_{2},\cdots, b_{N}$ be a basis of $\mathbb{F}_{2}^{N}$ such that $b_{1},\cdots b_{k_{\ell}}$ span $C_{\ell}$. The lattice $L$ consists of all vectors of the form
\begin{eqnarray}
\sum_{\ell=1}^{r-1}2^{\ell-1}\sum_{j=1}^{k_{\ell}}\alpha_{j}^{(\ell)}\psi(b_{j})+2^{r-1}l, \notag\
\end{eqnarray}
where $\alpha_{j}^{(\ell)}\in\{0,1\}$ and $l\in\mathbb{Z}^{N}$. The fundamental volume of a lattice obtained in this way is given by
\begin{eqnarray}
V(L)=2^{rN-\sum\limits_{\ell=1}^{r}k_{\ell}}.
\label{eqn:volume}
\end{eqnarray}
Barnes-Wall lattices are constructed from Reed-Muller codes by Construction D. For example, the code formula\footnote{From the definition of Construction D we know that not all the lattices constructed by Construction D have such code formulas \cite{oggier}. However from the decoder's point there is still a binary code in each level due to the mod 2 operation. } of the 1024-dimensional Barens-Wall lattice is:
\begin{eqnarray}
BW_{1024}=\text{RM}(1,10)+2\text{RM}(3,10)+\cdot\cdot\cdot+2^{5}\mathbb{Z}^{1024}.
\label{eqn:BW1024}
\end{eqnarray}

\section{Polar Codes for Mod-2 BAWGN Channel}
\label{sec:polarcodes}
In this section, we will show how to construct a polar code for each level in the multilevel construction of polar lattices. The $\mathbb{Z}/2\mathbb{Z}$ channel is in fact a mod-2 BAWGN channel. The signal flow of the mod-2 BAWGN channel is shown in Fig. \ref{fig:mod2signal}. The mod-2 operation is applied within $[-1,1]$ not $[0,2]$.

\begin{figure}
    \centering
    \includegraphics[width=8.5cm]{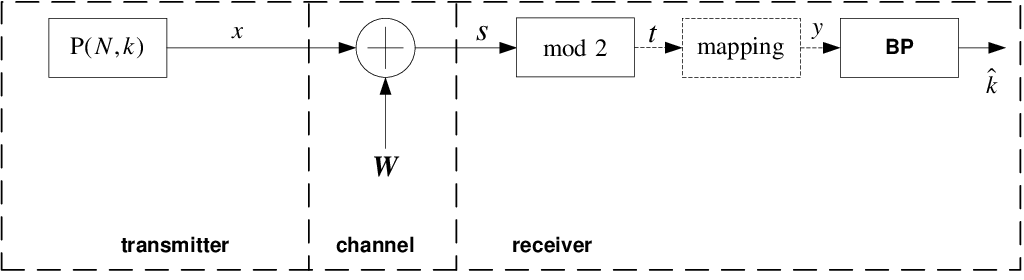}
    \caption{Signal flow of the mod-2 BAWGN channel. $P(N,k)$ represents the polar code with block length $N$ and $k$ information bits.}
    \vspace{-12pt}
    \label{fig:mod2signal}
\end{figure}

\subsection{Channel Capacity of Mod-2 BAWGN Channel}
In \cite{forney6}, the capacity of the mod-$\Lambda$ channel and $\Lambda/\Lambda'$ channel was derived. Therefore we can calculate the capacity of the mod-2 BAWGN channel.

Let $f(x)$ be the probability density function (PDF) of the Gaussian noise $x\in \mathbb{R}^{n}$
\begin{eqnarray}
f_{\sigma^{2}}(x)=(2\pi \sigma^{2})^{-n/2} e^{-\frac{\|x\|^{2}}{2\sigma^{2}}}.
\notag\
\end{eqnarray}
Then after the mod-$2\mathbb{Z}$ (mod-2) operation we have the density
\begin{eqnarray}
f_{2\mathbb{Z},\sigma^{2}}(x')=\sum_{i\in 2\mathbb{Z}}f_{\sigma^{2}}(x'+i), \quad x'\in [-1,1].
\label{eqn:mod2pdf}
\end{eqnarray}
The capacity of the mod-$\Lambda$ channel is
\begin{eqnarray}
C(\Lambda, \sigma^{2})=\text{log } V(\Lambda)-h(\Lambda, \sigma^{2}),
\label{eqn:modchannel}
\end{eqnarray}
where $h(\Lambda, \sigma^{2})$ is the differential entropy of the $\Lambda$-aliased noise
\begin{eqnarray}
h(\Lambda,\sigma^{2})=-\int_{\mathcal{V}(\Lambda)}f_{\Lambda,\sigma^{2}}(x')\text{ log } f_{\Lambda,\sigma^{2}}(x')dx' \notag\
\end{eqnarray}
where $\mathcal{V}(\Lambda)$ denotes the Voronoi region.

Furthermore, the capacity of the $\mathbb{Z}/2\mathbb{Z}$ channel is \cite{forney6}
\begin{eqnarray}
C(\mathbb{Z}/2\mathbb{Z})=C(2\mathbb{Z}, \sigma^{2})-C(\mathbb{Z}, \sigma^{2}). \notag\
\end{eqnarray}

The channel capacity of the $\mathbb{Z}/2\mathbb{Z}$ channel is shown in Fig. \ref{fig:channelcapacity}. Also shown are $2\mathbb{Z}/4\mathbb{Z}$ and $4\mathbb{Z}/8\mathbb{Z}$ channels which are also mod-2 BAWGN channels but upgraded versions (with smaller noise variance) of $\mathbb{Z}/2\mathbb{Z}$. For example, coding over $2\mathbb{Z}/4\mathbb{Z}$ is nothing special compared with coding over $\mathbb{Z}/2\mathbb{Z}$. The only difference is the signal power is increased by 4. Therefore it is equivalent to a $\mathbb{Z}/2\mathbb{Z}$ channel with the noise variance divided by 4. This structure simplifies our task to find good polar codes for each channel. We just focus on the $\mathbb{Z}/2\mathbb{Z}$ channel with different noise variances. The above finding also helps us to prove the nested structure of polar codes which are constructed for these mod-2 BAWGN channels.
\begin{figure}
    \centering
    \includegraphics[width=9cm]{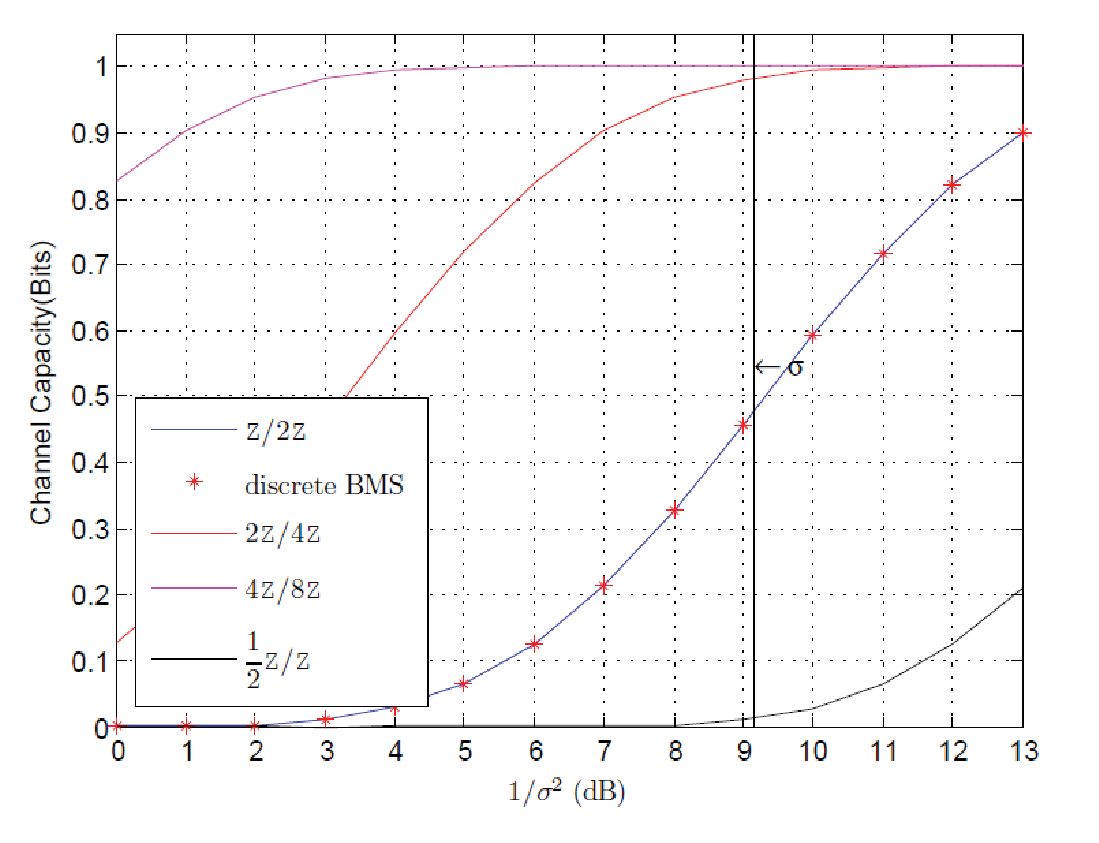}
    \vspace{-12pt}
    \caption{Channel capacity of the mod-2 BAWGN channel. The capacity of the discrete BSM channel is calculated in Sect. \ref{seq:constructionpolar} to show the negligible difference between the continuous mod-2 BAWGN channel and the quantized discrete channel.}
    \vspace{-12pt}
    \label{fig:channelcapacity}
\end{figure}

\subsection{Symmetry of Mod-2 BAWGN Channel}
Polar codes are proved to be able to achieve the capacity of output-symmetric channels. Therefore we need to show this mod-2 BAWGN channel is indeed an output-symmetric channel. The output of the mapping operation in Fig. \ref{fig:mod2signal} is
\begin{eqnarray}
y=2|t|-1,
\notag\
\end{eqnarray}
where the conditional PDF of $t$ can be derived from \eqref{eqn:mod2pdf}.

Then the conditional PDFs of $y$ are
\begin{eqnarray}
\left\{\begin{aligned}
&f(y|{x}=1)=\frac{1}{\sqrt{2\pi\sigma^2}}\sum\limits_{j=-\infty}^{+\infty}\text{exp}\left[-\frac{(y-1+4
j)^2}{8\sigma^2}\right], \notag\\
&f(y|{x}=0)=\frac{1}{\sqrt{2\pi\sigma^2}}\sum\limits_{j=-\infty}^{+\infty}\text{exp}\left[-\frac{(y+1+4
j)^2}{8\sigma^2}\right]. \notag\
\end{aligned}\right.
\end{eqnarray}
Therefore the PDFs of the output satisfy
\begin{eqnarray}
f(y|0)=f(-y|1), \text{  for } \text{all } y\in [-1,1].
\notag\
\end{eqnarray}
Clearly, it is a binary-input
output-symmetric channel.

\subsection{Construction of Polar Codes for Mod-2 BAWGN channel}
\label{seq:constructionpolar}
Let $W(y|x)$ be a BMS channel. Polar codes are block codes of length $N=2^m$ bits $\{u_i\}_{i=1}^{N}$ with binary alphabet $\mathcal{X}$. Let $I(W)$ be the symmetric capacity of $W$. Given the rate $R<I(W)$, a set of $RN$ rows of the generator matrix $G=\left[\begin{smallmatrix}1&0\\1&1\end{smallmatrix}\right]^{\otimes m}$ are to be chosen. This gives an $N$-dimensional channel $W_{N}(y_{1}^{N}|u_{1}^{N})$. The channels seen by each bit can be defined by \cite{polarcodes}
\begin{eqnarray}
W_{N}^{(i)}(y_{1}^{N},u_{1}^{i-1}|u_{i})=\sum\limits_{u_{i+1}^{N}\in \mathcal{X}^{N-i}}\frac{1}{2^{N-1}W_{N}(y_{1}^{N}|u_{1}^{N})}. \notag\
\end{eqnarray}
Ar{\i}kan proved that as $N$ grows $W_{N}^{(i)}$ approaches either an error-free channel or a completely noisy channel. The set of completely noisy (resp. error-free) bit channels is called frozen set $\mathcal{A}^{\mathcal{C}}$ (resp. free set $\mathcal{A}$). The corresponding bits in $\{u_i\}_{i=1}^{N}$ are called frozen bits (resp. free bits). Set $u_{i}=0$ for $i\in \mathcal{A}^{\mathcal{C}}$ and only send information bits within $\mathcal{A}$.

The block error probability $P_{B}$ can be upper bounded by the sum of the bit channels' Bhattacharyya parameters as $\overline{P}_{B}(W_{N})\leq\Sigma_{i\in \mathcal{A}}Z(W_{N}^{(i)})$. Constructing polar codes is equivalent
to choosing all the good indices. However, the complexity of the exact calculation appears to
be exponential in the block length. The authors in \cite{Ido,polarconstruction} proposed a practical quantization method to control the complexity. We use their method to construct polar codes for the mod-2 BAWGN channel. The details are given in \cite{Yan}.

We calculate the capacity of the quantized discrete channel which is generated by the Algorithm 1 in \cite{Yan}. We set $K=32$ so that $2K$ is the number of quantization intervals. Note that $1-I(W)$ is the expectation of the function $-x \text{ log }(x)-(1-x)\text{ log }(1-x)$ over the resulted distribution $P_{\chi}$. The results are shown in Fig. \ref{fig:channelcapacity}. It is seen that the capacity of this discrete BMS channel is almost the same as the theoretic capacity of the mod-2 BAWGN channel.

The approximation loss in the capacity due to the above quantization can be bounded in a precise manner. It was shown that the merging algorithms of polar code construction in \cite{Ido} and \cite{polarconstruction} share the same bound on the approximation loss. Denote the upper bound on the probability of error for the resulting degraded bit channel $\tilde{W}_{N}^{(i)}$ by $\overline{P}_{e}(\tilde{W}_{N}^{(i)}, K)$.
Given constant $\beta<1/2$, let the rate loss $\epsilon_{\textrm{loss}}$ be defined by
\begin{eqnarray}
\frac{1}{N}\left|\{i:\overline{P}_{e}(\tilde{W}_{N}^{(i)}, K)<2^{-N^{\beta}}\}\right|= I(W)-\epsilon_{\textrm{loss}}. \notag\
\end{eqnarray}
Then, according to Theorem 16 in \cite{Ido}, $\liminf_{m\to \infty} \epsilon_{\textrm{loss}}=\epsilon(W,K,\beta)$ which is not a function of $N$, and
\begin{eqnarray}
\lim_{K\rightarrow \infty}\epsilon(W,K,\beta)=0.
\label{eqn:rateloss}
\end{eqnarray}

It means we can get arbitrarily close to the optimal construction of a polar code if $K$ is sufficiently large.

\section{Sphere-Bound-Achieving Polar Lattices}
\label{sec:polarlattice}
In this section, we will prove the polar lattices can achieve the sphere bound. At first we will prove that the polar codes used in all levels are nested. Then performance bounds are derived, and polar lattices are proved to achieve the sphere bound asymptotically.

\subsection{Nested Polar Codes}
Assume $P(N,k_{\ell})$ $(1\leq\ell\leq r-1)$ is the polar code constructed at the $\ell$-th level. We select the channel indices according to the Bhattacharyya parameters. The following lemma shows that the polar codes in the multilevel construction are nested. This is a requirement of Construction D.

\begin{lemma}
Polar codes constructed in the multilevel construction are nested; i.e., $P(N,k_{1})\subseteq P(N,k_{2})\subseteq\cdot\cdot\cdot\subseteq P(N,k_{r-1})$.
\end{lemma}

\begin{proof}
By Lemma 4.7 in \cite{polarchannelandsource}, if a BMS channel $\tilde{W}$ is a degraded version of $W$, then $\tilde{W}_{N}^{(i)}$ is degraded with respect to $W_{N}^{(i)}$ and $Z(\tilde{W}_{N}^{(i)})\geq Z(W_{N}^{(i)})$. Let the target error probability be $P_{e}=2^{-N^{\beta}}$. The codewords are generated by $x^{N}=u_{\mathcal{A}}G_{\mathcal{A}}$, where $G_{\mathcal{A}}$ is the submatrix of $G$ formed by rows with indices in the free set $\mathcal{A}$. The free sets for these two channels are
\begin{eqnarray}
\left\{\begin{aligned}
&\mathcal{A}&=&\{i:Z(W_{N}^{(i)})<P_{e}\}, \\ \notag\
&\mathcal{\tilde{A}}&=&\{i:Z(\tilde{W}_{N}^{(i)})<P_{e}\}. \notag\
\end{aligned}\right.
\end{eqnarray}

Due to the fact that $Z(\tilde{W}_{N}^{(i)})\geq Z(W_{N}^{(i)})$, $\mathcal{\tilde{A}}\subseteq \mathcal{A}$. If we construct polar codes $P(N,\mathcal{A})$ over $W$ and $\tilde{P}(N,\mathcal{\tilde{A}})$ over $\tilde{W}$, $G_{\mathcal{\tilde{A}}}$ is a submatrix of $G_{\mathcal{A}}$. Therefore $\tilde{P}(N,\mathcal{A})\subseteq P(N,\mathcal{\tilde{A}})$.

For the one-dimensional lattice partition, the $2^{\ell-1}\mathbb{Z}/2^{\ell}\mathbb{Z}$ channel with noise variance $\sigma^{2}$ is equivalent to the $2^{\ell}\mathbb{Z}/2^{\ell+1}\mathbb{Z}$ channel with noise variance $4\sigma^{2}$ for $1\leq \ell \leq r$. Therefore the channel of the $\ell$-th level is degraded with respect to the channel of the $(\ell+1)$-th level, and consequently, $P(N,k_{\ell})\subseteq P(N,k_{\ell+1})$.
\end{proof}

\subsection{Performance Bounds}
As we use a multistage decoding algorithm, the overall (block) error probability $P_{e}(L, \sigma^{2})$ of a polar lattice is upper-bounded by the sum of the block error probabilities at individual levels, by the union bound. The block error probability of the polar code for each level is denoted by $P_{B}(\mathcal{C}_{\ell}, \sigma^{2})$ for $(1\leq \ell\leq r-1)$ and by $P_{B}(\Lambda_{r}^{N}, \sigma^{2})$ for the $r$-th level. Then
\begin{eqnarray}
P_{e}(L, \sigma^{2})&\leq& \sum_{\ell=1}^{r-1}P_{B}(\mathcal{C}_{\ell}, \sigma^{2})+P_{B}(\Lambda_{r}^{N}, \sigma^{2}),
\label{eqn:unionbound}
\end{eqnarray}
where the error probability of $\Lambda_{r}^{N}$ is given by
\begin{eqnarray}
P_{B}(\Lambda_{r}^{N},\sigma^{2})=1-\int_{\mathcal{V}(\Lambda_{r}^{N})}f_{\sigma^{2}}(x)dx.
\label{eqn:erroroflambda}
\end{eqnarray}

As shown in the previous subsection, we can construct nested polar codes with rates $\kappa(\mathcal{C}_{\ell})=\frac{k_{\ell}}{N}=C(\Lambda_{\ell}/\Lambda_{\ell+1}, \sigma^{2})-\epsilon_{\textrm{loss}}(\ell)$ such that the block error probability in each level $P_{B}(\mathcal{C}_{\ell}, \sigma^{2})$ is upper bounded by $2^{-N^{\beta_{\ell}}}$, for any $\beta_{\ell}<\frac{1}{2}$.

The volume of the lattices constructed by Construction D is
\begin{eqnarray}
V(L)=2^{-N\sum\limits_{\ell=1}^{r-1}\kappa_{\ell}}V(\Lambda_{r})^{N}.
\label{eqn:volumeofL}
\end{eqnarray}
Therefore the logarithmic VNR of $L$ is
\begin{eqnarray}
\log_2\alpha^2(L,\sigma^2)
&=&\log_2\frac{V(L)^\frac{2}{N}}{2\pi e\sigma^2} \\
&=&\log_2\frac{2^{-2\sum\limits_{\ell=1}^{r-1}\kappa_{\ell}}V(\Lambda_{r})^{2}}{2\pi e\sigma^2} \notag\\
&=&-2\sum\limits_{\ell=1}^{r-1}\kappa_{\ell}+2\log_2 V(\Lambda_{r})- \log_22\pi e\sigma^2. \notag\
\label{eqn:VNRwithe}
\end{eqnarray}
Define
\begin{eqnarray}
\left\{\begin{aligned}
&\epsilon_{1}=C(\Lambda_{1},\sigma^2), \notag\\
&\epsilon_{2}=h(\sigma^2)-h(\Lambda_{r},\sigma^2), \notag\\
&\epsilon_{3}=C(\Lambda_{1}/\Lambda_{r}, \sigma^{2})-\sum\limits_{\ell=1}^{r-1}\kappa_{\ell}, \notag\
\end{aligned}\right.
\end{eqnarray}
where $h(\sigma^2)=\frac{1}{2}\log_22\pi e\sigma^2$ which is the differential entropy of the Gaussian noise.

Then we have
\begin{eqnarray}
\log_2\alpha^2(L,\sigma^2)
&=&2(\log_2V(\Lambda_{r})-C(\Lambda_{1}/\Lambda_{r}, \sigma^{2})\notag\\ &+& \epsilon_{3}-\frac{1}{2}\log_22\pi e\sigma^2) \notag\\
&=&2(\log_2V(\Lambda_{r})-\log_2V(\Lambda_{r})+h(\Lambda_{r},\sigma^2)\notag\\\ &+&\epsilon_{1}+ \epsilon_{3}-\frac{1}{2}\log_22\pi e\sigma^2) \notag\\
&=&2(\epsilon_{1}-\epsilon_{2}+\epsilon_{3}). \notag\
\end{eqnarray}

It was shown in \cite{forney6} that $\epsilon_{2}\approx\pi P_e(\Lambda_{r},\sigma^2)$, which is negligible compared to the other two terms. We can drop the term $\epsilon_{2}$ such as to obtain an upper bound
\begin{eqnarray}
\log_2\alpha^2(L,\sigma^2)
\leq 2(\epsilon_{1}+\epsilon_{3}).
\label{eqn:minimumVNR}
\end{eqnarray}

From \cite{cong2} we can derive the bound $\epsilon_{1}\leq\epsilon_{\Lambda_{1}}(\sigma)\log_2 e$ where $\epsilon_{\Lambda_{1}}(\sigma)$ is the flatness factor. We may choose the lattice $\Lambda_{1}$ sufficiently fine so that $\epsilon_{\Lambda_{1}}(\sigma)$ is negligible.

By the union bound, we may choose $V(\Lambda_{r})$ large enough that $P_{B}(\Lambda_{r}^{N}, \sigma^{2})\approx0$.

From \eqref{eqn:rateloss}, $\epsilon_{3}$ is the sum of the rate losses
\begin{eqnarray}
\epsilon_{3}=\sum\limits_{\ell=1}^{r-1}\epsilon_{\textrm{loss}}({\ell}). \notag\
\end{eqnarray}
We have
\begin{eqnarray}
\lim_{K\rightarrow \infty}\lim_{N\rightarrow \infty}\epsilon_{3}=\lim_{K\rightarrow \infty}\sum\limits_{\ell=1}^{r-1}\epsilon(\Lambda_{\ell}/\Lambda_{\ell+1},K,\beta_{\ell})=0, \notag\
\end{eqnarray}
and the sum error probability of polar codes is upper bounded by
\begin{eqnarray}
\sum_{\ell=1}^{r-1}P_{B}(\mathcal{C}_{\ell}, \sigma^{2})\leq \sum_{\ell=1}^{r-1}N2^{-N^{\beta_{\ell}}} \notag\
\end{eqnarray}
which can be made arbitrarily small by increasing the block length of polar codes $N$.

In conclusion, we have the following theorem:
\begin{theorem}
The error probability of polar lattice $L$ is bounded by
\begin{eqnarray}
P_{e}(L, \sigma^{2})\leq \sum_{\ell=1}^{r-1}N2^{-N^{\beta_{\ell}}}+N\left(1-\int_{\mathcal{V}(\Lambda_{r})}f_{\sigma^{2}}(x)dx\right),
\label{eqn:errorbound}
\end{eqnarray}
with the logarithmic VNR bounded by (\ref{eqn:minimumVNR}). In the limit as $r,K,N\to \infty$, $L$ can achieve the sphere bound, i.e.,
$\alpha^2(L,\sigma^2)\rightarrow1$
and
$P_{e}(L, \sigma^{2})\rightarrow0$.
\end{theorem}

\begin{remark}
Note that the integral in (\ref{eqn:errorbound}) can simply be represented by the Q-function, since $\Lambda_r$ is a scaled version of $\mathbb{Z}$.
\end{remark}

\section{Practical Design and Simulation Results}
\label{sec:example}
We give a design example of polar lattices in this section. The design follows the capacity rule. We use the one-dimensional lattice partition $\mathbb{Z}/2\mathbb{Z}/\cdot\cdot\cdot/2^{r-1}\mathbb{Z}$.
We suggest using $\sigma$ not VNR
as the reference parameter to choose component codes. This is because $\sigma$ can be calculated in each level and can fully represent the channel.

One can determine the number of valid levels according to the capacity calculation as shown in Fig. \ref{fig:channelcapacity}. If the actual rate of a
level is almost 1 and still can achieve the target block error probability $P_{B}$ for the reference $\sigma$, it is unnecessary. This is because
the rate from this level will be canceled by $\frac{1}{2}\log_2 2\pi e \sigma_{1}^{2}$. The effective levels are those with an actual rate visibly within $(0,1)$. In other words, adding levels whose capacities are close to 1 or 0 do not noticeably change the VNR. Therefore, we choose two levels, which was also suggested in \cite{forney6}.


The multilevel construction and the multistage decoding are shown in Fig. \ref{fig:1D}. For the $i$-th level, $\alpha^{(i)} $ are information bits, $b_{1},b_{2},\cdots,b_{k_{i}}$ are a set of basis which are chosen from the matrix $G_{N}=\left[\begin{smallmatrix}1&0\\1&1\end{smallmatrix}\right]^{\otimes
m}$ according to the polarization rule for the $i$-th level's channel, $\sigma_{i}$
means the standard deviation. Note the difference of $\sigma$ between each level is only a factor of 2.

\begin{figure}
    \centering
    \includegraphics[width=8cm]{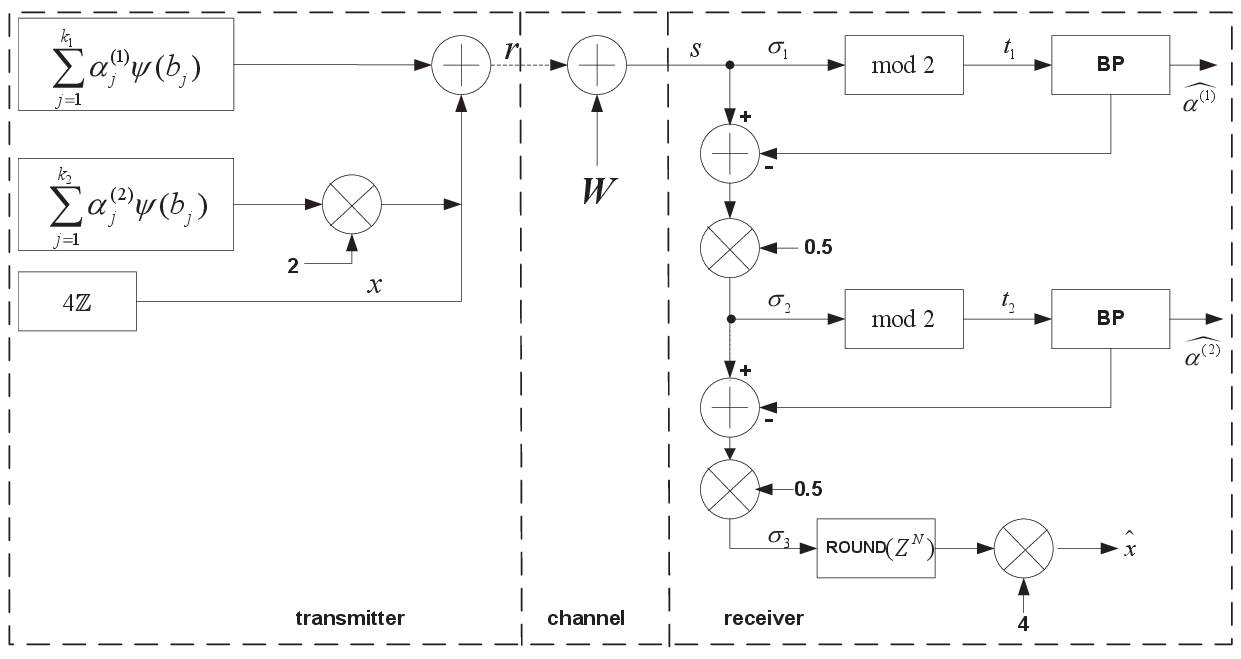}
    \caption{A polar lattice with two levels.}
    \vspace{-12pt}
    \label{fig:1D}
\end{figure}

We start from the bottom level to the top level. The bottom level is
a $\mathbb{Z}^{N}$ lattice decoder.
First, from \eqref{eqn:erroroflambda}, if we set the target error probability to $10^{-5}$, $\sigma_{3}\approx 0.08719$. Then, move to the middle level. $\sigma_{2}=2*\sigma_{3}=0.1744$. From Fig. \ref{fig:channelcapacity}, the channel capacity of the middle
level is $C(\mathbb{Z}/2\mathbb{Z},\sigma_{2})=0.9830$. Finally, the capacity of the top level is 0.4775 as $\sigma_{1}=0.3488$. Our goal is to find two polar codes with rates close to
their respective capacities and block error probabilities $\leq 10^{-5}$ over the mod-2 BAWGN channel.

From \eqref{eqn:VNRwithe},
\begin{eqnarray}
\log_2\alpha^2(L,\sigma^2)=2\left(\log_2 V(\Lambda_{r})-R-\frac{1}{2}\log_2 2\pi e \sigma_{1}^{2}\right),
\label{eqn:gap}
\end{eqnarray}
where $R$ is the sum rate of the component polar codes.

From our simulations, for $N=1024$, we
found that $P_{B}(\frac{k_{1}}{N}=0.22, \sigma_{1}^{2})\approx 2\times10^{-5}$, $P_{B}(\frac{k_{2}}{N}=0.9, \sigma_{2}^{2})\approx 2\times10^{-5}$ and
$P_{e}(\Lambda_{r}^{N}, \sigma_{3}^{2}) \approx 1\times10^{-5}$. Therefore $R=0.22+0.9$.
According to \eqref{eqn:unionbound} and \eqref{eqn:gap},
\begin{eqnarray}
P_{e}(L,\sigma_{1}^{2})\approx 5\times10^{-5},
\end{eqnarray}
and
\begin{eqnarray}
10\log_{10}\alpha^2(L,\sigma_{1}^{2})=2.12 \text{ dB},
\end{eqnarray}
which is not surprisingly very close to our simulation
shown in Fig. \ref{fig:1D_bw}. Furthermore, we also demonstrate the performance of polar lattice with $N=8192$ and $R=0.33+0.956$ in Fig. \ref{fig:1D_bw}.
This simulation implies that the performance of the component codes is very important to the multilevel
lattice codes. The gap to the sphere bound is largely due to the rate loss of polar codes.

Now we go back to the Barnes-Wall lattices. According to Fig. \ref{fig:channelcapacity}, there are only 2 valid levels, but the Barnes-Wall lattice \eqref{eqn:BW1024} has 6 levels for $N=1024$. The rate of the first level is $0.01$, which is much bigger than the capacity of the first level when VNR is low. Another reason is the relatively weak error-correction ability of Reed-Muller codes. Therefore the error probability of the first level will be very high in the low VNR region as shown in Fig. \ref{fig:1D_bw}.

\section{Conclusions}
In this paper, we have followed Forney \textit{et al.}'s multilevel approach to construct polar lattices. The channel capacity of each level is calculated and a polar code is constructed to achieve its capacity accordingly. Since polar lattices are as explicit as polar codes, their construction is equally efficient. It is worth mentioning that our polar lattice proposed in Section \ref{sec:example} has almost the same gap (1.45 dB) to the Poltyrev capacity as the 1000-dimensional LDLC \cite{ldlc} with the settings of \cite[Fig. 13]{finiteIC}. The theoretic minimum gap for the 1000-dimensional lattices is about 1 dB \cite{finiteIC}. Furthermore, compared with existing schemes \cite{ldpclattice,IntegerLDLC,ldlc,turbolattice}, polar lattices are distinguished by their provable AWGN-goodness and low complexity, namely, they asymptotically achieve the sphere bound with multistage decoding.


%



\section*{Acknowledgments}

The authors would like to thank Prof. Jean-Claude Belfiore for helpful discussions.
The work of Yanfei Yan is supported by the China Scholarship Council.




\bibliographystyle{IEEEtran}
\bibliography{yanfei}

%








\end{document}